\documentclass[%
reprint,
nofootinbib,
 amsmath,amssymb,
 aps,
 pra,
]{revtex4-2}

\usepackage{mathrsfs}

\usepackage{graphicx}
\usepackage{dcolumn}
\usepackage{bm}

\usepackage{amsthm}
\usepackage{xfrac}
\usepackage[mathcal]{euscript}

\usepackage{url}
\usepackage{hyperref}
\usepackage{color}
\definecolor{refcolor}{RGB}{0,0,190}
\hypersetup{
    colorlinks,
    citecolor=refcolor,
    filecolor=refcolor,
    linkcolor=refcolor,
    urlcolor=refcolor
}
\usepackage{bookmark}
\bookmarksetup{
  numbered, 
  open,
}

\newtheorem{theorem}{Theorem}

\newtheorem{defCustom}{}
\makeatletter
\newcommand{\setdefCustomtag}[1]{
  \let\oldthedefCustom\thedefCustom
  \renewcommand{\thedefCustom}{\texttt{#1}}
  \g@addto@macro\enddefCustom{
    \global\let\thedefCustom\oldthedefCustom}
  }
\makeatother

\theoremstyle{remark}

\theoremstyle{definition}

\newtheorem{condition}{Condition}
\renewcommand{\thecondition}{\arabic{condition}}
\makeatletter
\newcommand{\setconditiontag}[1]{
  \let\oldthecondition\thecondition
  \renewcommand{\thecondition}{#1}
  \g@addto@macro\endcondition{
    \global\let\thecondition\oldthecondition}
  }
\makeatother

\newtheorem{definition}{Definition}

\theoremstyle{remark}

\newtheorem{example}{Example}

\definecolor{greenPsi}{rgb}{0.0, 0.375, 0.0}
\definecolor{blueStruct}{rgb}{0.0, 0.0, 1.0}
\definecolor{redStruct}{rgb}{1.0, 0.0, 0.0}

\begin{document}


\def\({\left(}
\def\){\right)}
\def\[{\left[}
\def\]{\right]}

\newcommand{\tn}{\textnormal}
\newcommand{\ds}{\displaystyle}
\newcommand{\dsfrac}[2]{\displaystyle{\frac{#1}{#2}}}

\newcommand{\boplus}{\textstyle{\bigoplus}}
\newcommand{\botimes}{\textstyle{\bigotimes}}
\newcommand{\bcup}{\textstyle{\bigcup}}
\newcommand{\bsqcup}{\textstyle{\bigsqcup}}
\newcommand{\bcap}{\textstyle{\bigcap}}

\newcommand{\MQS}{\ref{def:MQS}}
\newcommand{\HSF}{\ref{def:HSF}}

\newcommand{\struct}{\mc{S}}
\newcommand{\kind}{\mc{K}}

\newcommand{\statespace}{\mathcal{S}}
\newcommand{\hilbert}{\mathcal{H}}
\newcommand{\vectorspace}{\mathcal{V}}
\newcommand{\mc}[1]{\mathcal{#1}}
\newcommand{\ms}[1]{\mathscr{#1}}

\newcommand{\wh}[1]{\widehat{#1}}
\newcommand{\dwh}[1]{\wh{\rule{0ex}{1.3ex}\smash{\wh{\hfill{#1}\,}}}}

\newcommand{\wt}[1]{\widetilde{#1}}
\newcommand{\wht}[1]{\widehat{\widetilde{#1}}}
\newcommand{\on}[1]{\operatorname{#1}}

\newcommand{\qmU}{$\mathscr{U}$}
\newcommand{\qmR}{$\mathscr{R}$}
\newcommand{\qmUR}{$\mathscr{UR}$}
\newcommand{\qmDR}{$\mathscr{DR}$}

\newcommand{\R}{\mathbb{R}}
\newcommand{\C}{\mathbb{C}}
\newcommand{\Z}{\mathbb{Z}}
\newcommand{\K}{\mathbb{K}}
\newcommand{\N}{\mathbb{N}}
\newcommand{\Prj}{\mathcal{P}}
\newcommand{\abs}[1]{\left|#1\right|}

\newcommand{\de}{\operatorname{d}}
\newcommand{\tr}{\operatorname{tr}}
\newcommand{\im}{\operatorname{Im}}

\newcommand{\ie}{\textit{i.e.}\ }
\newcommand{\vs}{\textit{vs.}\ }
\newcommand{\eg}{\textit{e.g.}}
\newcommand{\cf}{\textit{cf.}\ }
\newcommand{\etc}{\textit{etc}}
\newcommand{\etal}{\textit{et al.}}

\newcommand{\Span}{\tn{span}}
\newcommand{\pde}{PDE}
\newcommand{\U}{\tn{U}}
\newcommand{\SU}{\tn{SU}}
\newcommand{\GL}{\tn{GL}}

\newcommand{\schrod}{Schr\"odinger}
\newcommand{\vonneum}{Liouville-von Neumann}
\newcommand{\ks}{Kochen-Specker}
\newcommand{\leggarg}{Leggett-Garg}
\newcommand{\bra}[1]{\langle#1|}
\newcommand{\ket}[1]{|#1\rangle}
\newcommand{\braket}[2]{\langle#1|#2\rangle}
\newcommand{\ketbra}[2]{|#1\rangle\langle#2|}
\newcommand{\expectation}[1]{\langle#1\rangle}
\newcommand{\Herm}{\tn{Herm}}
\newcommand{\Sym}[1]{\tn{Sym}_{#1}}
\newcommand{\meanvalue}[2]{\langle{#1}\rangle_{#2}}

\newcommand{\btimes}{\boxtimes}
\newcommand{\btimess}{{\boxtimes_s}}

\newcommand{\h}{\mathbf{(2\pi\hbar)}}
\newcommand{\x}{\mathbf{x}}
\newcommand{\xThree}{\boldsymbol{x}}
\newcommand{\z}{\mathbf{z}}
\newcommand{\q}{\mathbf{q}}
\newcommand{\p}{\mathbf{p}}
\newcommand{\0}{\mathbf{0}}
\newcommand{\annih}{\widehat{\mathbf{a}}}

\newcommand{\cs}{\mathscr{C}}
\newcommand{\ps}{\mathscr{P}}
\newcommand{\xhat}{\widehat{\x}}
\newcommand{\phat}{\widehat{\mathbf{p}}}
\newcommand{\fqproj}[1]{\Pi_{#1}}
\newcommand{\cqproj}[1]{\wh{\Pi}_{#1}}
\newcommand{\cproj}[1]{\wh{\Pi}^{\perp}_{#1}}

\newcommand{\M}{\mathbb{E}_3}
\newcommand{\D}{\mathbf{D}}
\newcommand{\dn}{\tn{d}}
\newcommand{\db}{\mathbf{d}}
\newcommand{\n}{\mathbf{n}}
\newcommand{\m}{\mathbf{m}}
\newcommand{\V}[1]{\mathbb{V}_{#1}}
\newcommand{\F}[1]{\mathcal{F}_{#1}}
\newcommand{\Fvacuumfield}{\widetilde{\mathcal{F}}^0}
\newcommand{\nD}[1]{|{#1}|}
\newcommand{\Lin}{\mathcal{L}}
\newcommand{\End}{\tn{End}}
\newcommand{\vbundle}[4]{{#1}\to {#2} \stackrel{\pi_{#3}}{\to} {#4}}
\newcommand{\vbundlex}[1]{\vbundle{V_{#1}}{E_{#1}}{#1}{M_{#1}}}
\newcommand{\rep}{\rho_{\scriptscriptstyle\btimes}}

\newcommand{\intl}[1]{\int\limits_{#1}}

\newcommand{\moyalBracket}[1]{\{\mskip-5mu\{#1\}\mskip-5mu\}}

\newcommand{\Hint}{H_{\tn{int}}}

\newcommand{\quot}[1]{``#1''}

\def\sref #1{\S\ref{#1}}

\newcommand{\dBB}{de Broglie--Bohm}
\newcommand{\dBBt}{{\dBB} theory}
\newcommand{\pwt}{pilot-wave theory}
\newcommand{\PWT}{PWT}
\newcommand{\NRQM}{{\textbf{NRQM}}}

\newcommand{\image}[3]{
\begin{center}
\begin{figure}[!ht]
\includegraphics[width=#2\textwidth]{#1}
\caption{\small{\label{#1}#3}}
\end{figure}
\vspace{-0.35in}
\end{center}
}

\title{Refutation of Hilbert Space Fundamentalism}

\author{Ovidiu Cristinel Stoica}
\affiliation{
 Dept. of Theoretical Physics, NIPNE---HH, Bucharest, Romania. \\
	Email: \href{mailto:cristi.stoica@theory.nipne.ro}{cristi.stoica@theory.nipne.ro},  \href{mailto:holotronix@gmail.com}{holotronix@gmail.com}
	}%

\date{\today}

\begin{abstract}
According to the ``Hilbert Space Fundamentalism'' Thesis, all features of a physical system, including the $3$D-space, a preferred basis, and factorization into subsystems, uniquely emerge from the state vector and the Hamiltonian alone. I give a simplified account of the proof from \href{http://arxiv.org/abs/2102.08620}{arXiv:2102.08620} showing that such emerging structures cannot be both unique and physically relevant.
\end{abstract}


\maketitle

\subsection*{Hilbert Space Fundamentalism}

Quantum Mechanics (QM) represents the state of any closed system, which may even be the entire universe, as a vector $\ket{\psi(t)}$ in a complex \emph{Hilbert space} $\hilbert$. The \emph{Hamiltonian} operator $\wh{H}$ fully specifies its \emph{evolution equation},
\begin{equation}
	\label{eq:unitary_evolution}
	\ket{\psi(t)}=\wh{U}_{t,t_0}\ket{\psi(t_0)},
\end{equation}
where $\wh{U}_{t,t_0}:=e^{-i/\hbar(t-t_0)\wh{H}}$ is the \emph{time evolution operator}. 

Since QM is invariant to unitary symmetries, this seems to justify the following:
\setdefCustomtag{HSF}
\begin{defCustom}
\label{def:HSF}
The \emph{Hilbert-space fundamentalism Thesis (\HSF):} everything about a physical system, including the $3$D-space, a preferred basis, a preferred factorization of the Hilbert space (needed to represent subsystems, {\eg} particles), emerge uniquely from the triple
\begin{equation}
	\label{eq:MQS}
	(\hilbert,\wh{H},\ket{\psi}).
\end{equation}
\end{defCustom}

\setdefCustomtag{MQS}
\begin{defCustom}
	\label{def:MQS}
	We call $(\hilbert,\wh{H},\ket{\psi})$ \emph{minimalist quantum structure}.
\end{defCustom}

The {\HSF} Thesis is sometimes assumed more or less explicitly in some versions of various approaches to QM, like \emph{information-theoretic}, \emph{decoherence}, and \emph{Everett's Interpretation}. 
The sufficiency of the {\MQS} is claimed, perhaps most explicitly, by Carroll and Singh \cite{CarrollSingh2019MadDogEverettianism}, p. 95:
\begin{quote}
	Everything else--including space and fields propagating on it--is emergent from these minimal elements.
\end{quote}

But this goes beyond Everettianism. We read in \cite{Carroll2021RealityAsAVectorInHilbertSpace}
\begin{quote}
The laws of physics are determined solely by the energy eigenspectrum of the Hamiltonian. 
\end{quote} 

Scott Aaronson states in ``The Zen Anti-Interpretation of Quantum Mechanics'' \cite{Aaronson2021TheZenAntiInterpretationOfQuantumMechanics} that a quantum state is 
\begin{quote}
	a unit vector of complex numbers [...] which encodes everything there is to know about a physical system.
\end{quote} 

I confess that I think the {\HSF} Thesis makes sense, if we take seriously the unitary symmetry of QM.
Why would the position basis of the space $\hilbert$ be fundamental, when it's like other bases of $\hilbert$? Just like the reference frames provide  convenient descriptions of the $3$D-space with no physical reality, why would $\ket{\x}$ and $\ket{\p}$ be fundamental among all of the bases of $\hilbert$? And if they play a privileged role, shouldn't this role emerge from the {\MQS} alone?
Unfortunately, we will see that this cannot happen.

\subsection*{Refutation of Hilbert Space Fundamentalism}

In \cite{Sto2021SpaceThePreferredBasisCannotUniquelyEmergeFromTheQuantumStructure} I gave a fully general proof that the {\HSF} Thesis is not true: for any {\MQS}, no emerging structure can be both \emph{unique} and \emph{physically relevant}.
Here I will give a less abstract and less technical version of that proof.

Denote a preferred structure expected to uniquely emerge from the {\MQS} by $\struct_{\wh{H}}^{\ket{\psi}}$, to express its dependence on $\wh{H}$ and $\ket{\psi}$. Any such structure should be invariant, so it can be defined in terms of \emph{tensor objects} from $\botimes^r\hilbert\otimes\botimes^s\hilbert^\ast$ \cite{Weyl1946ClassicalGroupsInvariantsAndRepresentations}.

For a structure to be considered ``preferred'', it has to consist of tensors that satisfy specific conditions, expressed as invariant tensor equations or inequations.
For example,
a preferred basis can be given by a set of operators $\wh{A}_{\alpha}\in\hilbert\otimes\hilbert^\ast$, satisfying the relations $\wh{A}_{\alpha}\wh{A}_{\alpha'}-\wh{A}_{\alpha}\delta_{\alpha\alpha'}=0$,
$\wh{I}_{\hilbert}-\sum_\alpha\wh{A}_{\alpha}$=0, and
$\tr\wh{A}_{\alpha}=1$.

This is generally true, and justifies the following
\begin{definition}
\label{def:kind}
The \emph{kind} $\kind$ of a structure is given by the types of its tensors and the defining tensor (in)equations specific to that structure (more details in \cite{Sto2021SpaceThePreferredBasisCannotUniquelyEmergeFromTheQuantumStructure}).
\end{definition}

In every case of interest, a set of Hermitian operators $\struct_{\wh{H}}^{\ket{\psi}}=\big(\wh{A}_\alpha^{\ket{\psi}}\big)_{\alpha\in\mc{A}}$ suffices to specify the candidate preferred structure.
We already saw this for a preferred basis.
The $3$D-space can be given by the \emph{number operators} for particles of each type $P$ at each point $\xThree$ in space,  $\widehat{a}_P^\dagger(\xThree)\widehat{a}_P(\xThree)$.
Tensor product structures of $\hilbert$ can be specified in terms of local algebras of Hermitian operators.
Therefore, we can limit to Hermitian operators only, avoiding much of the mathematical framework developed in \cite{Sto2021SpaceThePreferredBasisCannotUniquelyEmergeFromTheQuantumStructure} to include any kind of tensor objects.

Two state vectors $\ket{\psi_1}$ and $\ket{\psi_2}$ are \emph{physically equivalent}, $\ket{\psi_1}\sim\ket{\psi_2}$ if they represent the same physical state, \ie if they are related by gauge symmetry, space isometries \etc.
There is a group $G_P$ of unitary or anti-unitary transformations of $\hilbert$, so that $\ket{\psi_1}\sim\ket{\psi_2}$ iff $\ket{\psi_2}=\wh{g}\ket{\psi_1}$ for some $\wh{g}\in G_P$.
This equivalence extends uniquely to tensors on $\hilbert$.
The main claim of the {\HSF} Thesis is:
\begin{condition}[``Essentially unique''-ness]
\label{cond:uniqueness}
For any two $\kind$-structures $\struct_{\wh{H}}^{\ket{\psi}}=(\wh{A}_\alpha^{\ket{\psi}})_{\alpha\in\mc{A}}$ and $\struct_{\wh{H}}^{'\ket{\psi}}=({\wh{A}'}{}_\alpha^{\ket{\psi}})_{\alpha\in\mc{A}}$, there is a symmetry transformation $\wh{g}\in G_P$ so that
\begin{equation}
\label{eq:uniqueness}
	\(\wh{A}_\alpha^{\ket{\psi}}\)_{\alpha\in\mc{A}} = \(\wh{g}{\wh{A}'}{}_\alpha^{\ket{\psi}}\wh{g}^\dagger\)_{\alpha\in\mc{A}}.
\end{equation}
\end{condition}

Often the symmetry transformation from eq. \eqref{eq:uniqueness} simply permutes the operators $\wh{A}_\alpha^{\ket{\psi}}$. For example, space isometries permute the position operators.

To be \emph{physically relevant}, a $3$D-space should be able to distinguish among physically distinct states, {\eg}, since the density $\bra{\psi}\widehat{a}_P^\dagger(\xThree)\widehat{a}_P(\xThree)\ket{\psi}$ can be different for different values of $\ket{\psi}$, the same should be true for any candidate preferred $3$D-space.
In particular, it should detect changes that the Hamiltonian cannot detect (the Hamiltonian cannot distinguish $\ket{\psi}$ from $\ket{\psi'}$ iff there is a unitary $\wh{S}$ so that $[\wh{H},\wh{S}]=0$ and  $\ket{\psi'}=\wh{S}\ket{\psi}$, for example $\wh{S}=\wh{U}_{t_1,t_0}$, but space should be able to distinguish them).
The same applies to the components of $\ket{\psi}$ in a preferred basis,
and to tensor factors of $\hilbert$.
Otherwise, such structures would have no physical relevance.

Physical equivalence requires that the structure has to distinguish not merely the unit vectors $\ket{\psi}\nsim\ket{\psi'}\in\hilbert$, but any other state vectors $\wh{g}\ket{\psi}$ and $\wh{g'}\ket{\psi'}\in\hilbert$ representing the same physical states, for any $\wh{g},\wh{g'}\in G_P$.

\begin{condition}[Physical relevance]
\label{cond:TS_physical_relevance}
There exist at least two unit vectors $\ket{\psi}\nsim\ket{\psi'}\in\hilbert$ representing distinct physical states not distinguished by the Hamiltonian, so that for any symmetry transformations $\wh{g},\wh{g'}\in G_P$,
\begin{equation}
\label{eq:TS_physical_relevance_invar_G}
\(\bra{\psi}\wh{g}^\dagger\struct_{\wh{H}}^{\ket{\psi}}\wh{g}\ket{\psi}\)_{\alpha\in\mc{A}}
\neq\(\bra{\psi'}\wh{g'}^\dagger\struct{}_{\wh{H}}^{\ket{\psi'}}\wh{g'}\ket{\psi'}\)_{\alpha\in\mc{A}}.
\end{equation}
\end{condition}

\image{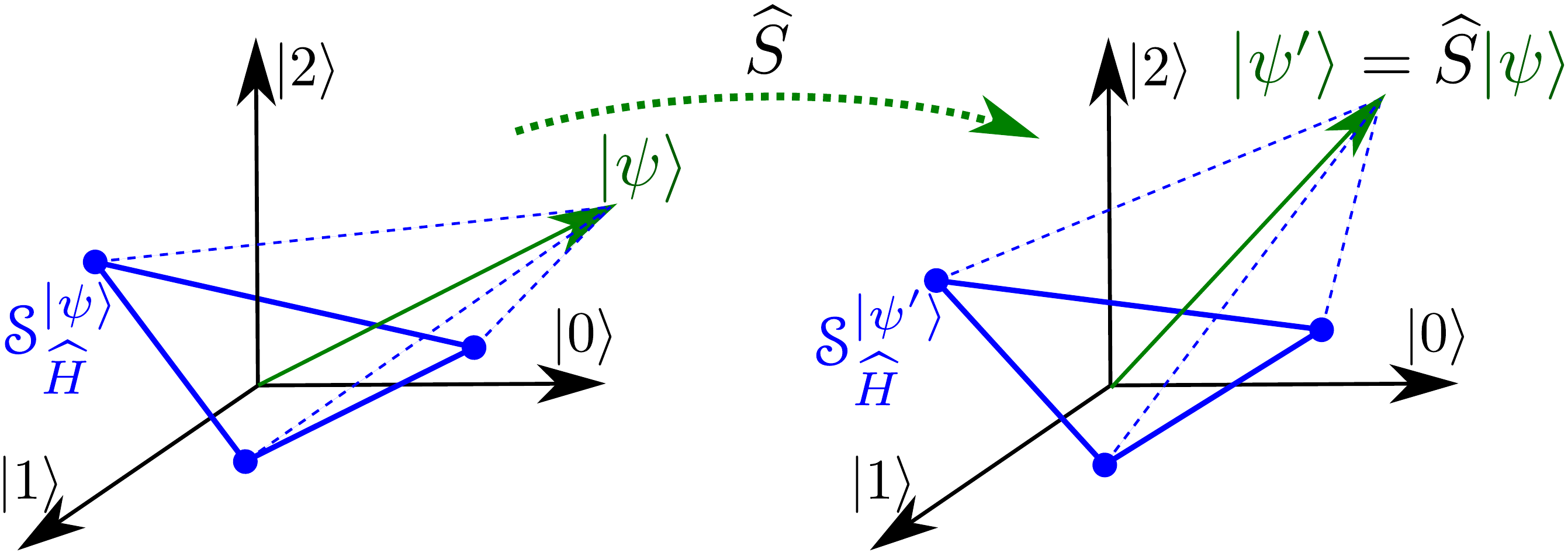}{0.45}{
Physical Relevance Condition \ref{cond:TS_physical_relevance}.
The candidate preferred structures $\textcolor{blueStruct}{\struct_{\wh{H}}^{\ket{\psi}}}=\(\wh{A}_\alpha^{\ket{\psi}}\)_{\alpha\in\mc{A}}$ and $\textcolor{blueStruct}{\struct_{\wh{H}}^{\ket{\psi'}}}=\(\wh{A}_\alpha^{\ket{\psi'}}\)_{\alpha\in\mc{A}}$ are symbolized as solid blue triangles. The dashed blue lines symbolize their relations with the corresponding vectors $\ket{\psi}$ and $\ket{\psi'}=\wh{S}\ket{\psi}$, \cf eq. \eqref{eq:TS_physical_relevance_invar_G}.
Condition \ref{cond:TS_physical_relevance} states that for any state vector $\ket{\psi}$ there are physically distinct state vectors $\ket{\psi'}\neq\ket{\psi}$ for which these relations are different.}

\begin{example}
\label{ex:time_distinguishing}
An obvious example satisfying Condition \ref{cond:TS_physical_relevance} is given by $\wh{S}=\wh{U}_{t_1,t_0}$, where $\ket{\psi(t_1)}=\wh{U}_{t_1,t_0}\ket{\psi(t_0)}$, because we expect that space, a preferred basis, and a preferred factorization to detect state changes in time.
\end{example}

I now prove that these two conditions are incompatible.

\begin{theorem}
	\label{thm:nogo}
	If a $\kind$-structure is physically relevant, then it is not the only $\kind$-structure.
\end{theorem}
\begin{proof}
Fig. \ref{nonuniqueness.pdf} illustrates the following construction.
The $\kind$-structure $\textcolor{blueStruct}{\struct_{\wh{H}}^{\ket{\psi}}}=\(\wh{A}_\alpha^{\ket{\psi}}\)_\alpha$ being a tensor object, its defining conditions are invariant to any unitary transformation $\wh{S}$ that commutes with $\wh{H}$. 
Then $\wh{S}^{-1}$  transforms $\textcolor{blueStruct}{\struct_{\wh{H}}^{\wh{S}\ket{\psi}}}$ into a $\kind$-structure $\textcolor{redStruct}{\struct_{\wh{H}}^{'\ket{\psi}}}=\wh{S}^{-1}\left[\textcolor{blueStruct}{\struct_{\wh{H}}^{\wh{S}\ket{\psi}}}\right]$ for $\ket{\psi}$,
\begin{equation}
	\label{eq:struct_unitary}
	\textcolor{redStruct}{\struct_{\wh{H}}^{'\ket{\psi}}}:=\(\wh{S}^{-1}\wh{A}_\alpha^{\wh{S}\ket{\psi}}\wh{S}\)_\alpha.
\end{equation}

From the uniqueness condition \eqref{eq:uniqueness}, \eqref{eq:struct_unitary} implies that
\begin{equation}
	\label{eq:uniqueness_unitary_inv}
	\(\wh{g}^{-1}\wh{A}_\alpha^{\ket{\psi}}\wh{g}\)_\alpha \stackrel{\eqref{eq:struct_unitary}}{=} \(\wh{S}^{-1}\wh{A}_\alpha^{\wh{S}\ket{\psi}}\wh{S}\)_\alpha
\end{equation}
for some transformation $\wh{g}\in G_P$.
From eq. \eqref{eq:uniqueness_unitary_inv}
\begin{equation}
	\label{eq:uniqueness_unitary_inv_psi}
\begin{aligned}
	\(\bra{\psi}\wh{g}^{-1}\wh{A}_\alpha^{\ket{\psi}}\wh{g}\ket{\psi}\)_{\alpha}
	&\stackrel{\eqref{eq:uniqueness_unitary_inv}}{=} \(\bra{\psi}\wh{S}^{-1}\wh{A}_\alpha^{\wh{S}\ket{\psi}}\wh{S}\ket{\psi}\)_{\alpha}\\
	&\ =\ \(\bra{\psi'}\wh{A}_\alpha^{\ket{\psi'}}\ket{\psi'}\)_\alpha.\\
\end{aligned}
\end{equation}

This contradicts Condition \ref{cond:TS_physical_relevance}.
\end{proof}

\image{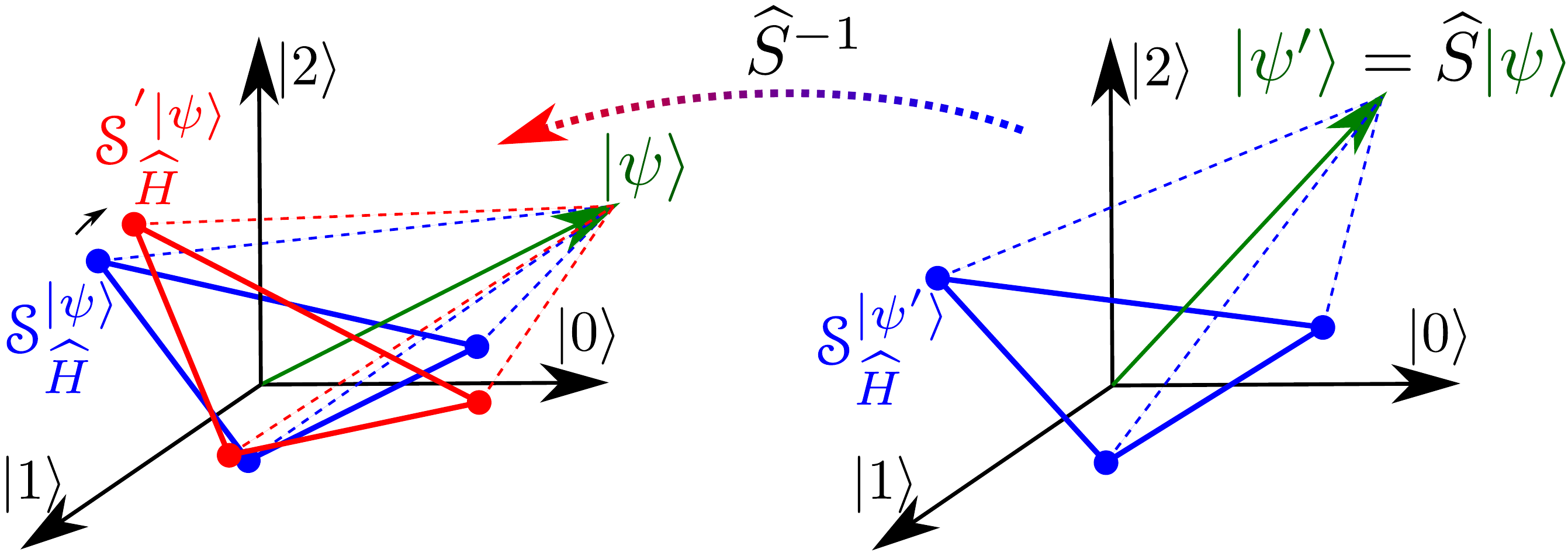}{0.45}{Construction of the 
$\kind$-structure $\textcolor{redStruct}{\struct_{\wh{H}}^{'\ket{\psi}}}=\wh{S}^{-1}\left[\textcolor{blueStruct}{\struct_{\wh{H}}^{\wh{S}\ket{\psi}}}\right]$ for $\ket{\psi}$ used in the proof of Theorem \ref{thm:nogo}.
}

Example \ref{ex:time_distinguishing} gives an infinite family of physically distinct ways to choose the $3$D-space. But in fact there is such an infinite family for each generator commuting with $\wh{H}$.

Theorem \ref{thm:nogo} was applied explicitly to various constructions assumed by the {\HSF} Thesis to emerge from the {\MQS}, and showed that, if they are physically relevant, they are not unique \cite{Sto2021SpaceThePreferredBasisCannotUniquelyEmergeFromTheQuantumStructure}. In particular, uniqueness fails for the generalized preferred basis (including for subsystems), the tensor product structure, emergent $3$D-space, emergent classicality \etc.
Such structures cannot emerge uniquely from the {\MQS}, a choice is always required.

\vspace{0.1in}

A question stands: what breaks the symmetry of QM?

\vspace{-0.15in}

\providecommand{\bysame}{\leavevmode\hbox to3em{\hrulefill}\thinspace}
\providecommand{\MR}{\relax\ifhmode\unskip\space\fi MR }
\providecommand{\MRhref}[2]{%
  \href{http://www.ams.org/mathscinet-getitem?mr=#1}{#2}
}
\providecommand{\href}[2]{#2}

\end{document}